\documentclass[a4]{article}
\usepackage[utf8]{inputenc} 
\usepackage{amsmath, amsfonts, amsthm, amssymb}
\usepackage{url}
\usepackage{graphicx}
\usepackage{color}
\usepackage{authblk}

\makeatletter
\let\@fnsymbol\@arabic
\makeatother

\newcommand{\R}{\mathbf{R}}

\newcommand{\C}{\mathbf{C}}
\newcommand{\Z}{\mathbf{Z}}

\newcommand{\D}{\mathcal{D}}
\newcommand{\U}{\mathcal{U}}
\newcommand{\EV}[1]{\mathbb{E}_{#1}\,}

\newcommand{\eps}{\varepsilon}
\newcommand{\be}{\mathbf{e}}
\newcommand{\bu}{\mathbf{u}}
\newcommand{\bx}{\mathbf{x}}
\newcommand{\bX}{\mathbf{X}}
\newcommand{\by}{\mathbf{y}}
\newcommand{\bz}{\mathbf{z}}
\newcommand{\e}{\mathrm{e}}

\renewcommand{\i}{\mathrm{i}}

\DeclareMathOperator{\sgn}{sgn}
\DeclareMathOperator{\supp}{supp}
\DeclareMathOperator{\Ent}{Ent}

\newtheorem{theorem}{Theorem}[section]
\newtheorem{lemma}[theorem]{Lemma}
\newtheorem{proposition}[theorem]{Proposition}
\newtheorem{corollary}[theorem]{Corollary}
\newtheorem{remark}[theorem]{Remark}
\newtheorem{definition}[theorem]{Definition}

\title{On the Theorem of Uniform Recovery of Random Sampling Matrices}
\author{Joel Andersson\thanks{Corresponding author. E-mail: joelan@kth.se, Phone: +4687906196} }
\author{Jan-Olov Str\"omberg\thanks{E-mail: jostromb@kth.se, Phone: +4687906676}}
\affil{Department of Mathematics, KTH, SE-100 44, Stockholm, Sweden}
\date{}

\begin{document}
\maketitle
\begin{abstract}
We consider two theorems from the theory of compressive sensing. Mainly a theorem concerning uniform recovery of random sampling matrices, where the number of samples needed in order to recover an $s$-sparse signal from linear measurements (with high probability) is known to be $m\gtrsim s(\ln s)^3\ln N$. We present new and improved constants together with what we consider to be a more explicit proof. A proof that also allows for a slightly larger class of $m\times N$-matrices, by considering what we call \emph{low entropy}. We also present an improved condition on the so-called restricted isometry constants, $\delta_s$, ensuring sparse recovery via $\ell^1$-minimization. We show that $\delta_{2s}<4/\sqrt{41}$ is sufficient and that this can be improved further to almost allow for a sufficient condition of the type $\delta_{2s}<2/3$.
\end{abstract}
\section{Introduction}
The theory of compressive sensing has emerged over the last 6-8 years, with the results we will consider originally presented by Tao, Candès et.al. in \cite{Tao2} and \cite{Tao1}. Rudelson and Vershynin improved the results in \cite{RudVer} and further generalizations where made by Rauhut in \cite{Rau1}, which also offers a nice overview of the topic. Today there is a vast literature on the topic of which the authors would also like to mention also \cite{Don1} and \cite{Can2}. Spanning a wide range of results, we do not aim to do a rigorous overview here but instead refers to mentioned papers from where we have gathered a lot of inspiration and where many further references can be found.\\ 
The beginning of section \ref{sec:prel} provides only a brief introduction to the topic with concepts that should be familiar to those that have encountered compressive sensing before. At the end of the section we present an improved version of a theorem from \cite{MLric}, regarding when the \emph{restricted isometry property} implies the \emph{null space property}.\\ 
In section \ref{sec:lem} the most important inequalities and lemmas, to be used in the proof of the main results of section \ref{sec:mains}, is presented. This section could possibly be skipped by readers familiar with the topic.\\
Our main concern will be the theorem of uniform recovery for random sampling matrices. To our knowledge the best result known is due to Cheraghchi, Guruswami and Velingker in \cite{CGVpp}. The theorem is stated to hold for the special case of a discrete Fourier matrix, but the authors remark that it also goes through for bounded orthonormal matrices. The result is the best in terms of asymptotics, and we will re-use a lot of their arguments but also provide constants that are improved compared with earlier results that we have encountered. We feel that our proof is more explicit in some ways, which we hope can offer more understanding of the techniques. First, in section \ref{sec:cmp}, we go into more detail about the differences and similarities of our work compared to the other mentioned ones.

\section{Comparisons with previous results}
\label{sec:cmp}
In \cite{CGVpp}, the following version of theorem \ref{thm:main1} is proved (using our notations and terminology):
\begin{theorem}[\cite{CGVpp}, Theorem 19]
\label{cgv}
Let $A\in\C^{m\times N}$ be an orthonormal matrix with entries bounded by $O(1/\sqrt N)$. Then for every $\delta,\epsilon>0$ and $N>N_0(\delta,\epsilon)$, with probability at least $1-\epsilon$ the restricted isometry constants $\delta_s$ of $\sqrt{N/m}A$ are less than $\delta$ for some $m$ satisfying
\[
m\lesssim\frac{\ln(1/\epsilon)}{\delta^2}s(\ln s)^3\ln N.
\]
\end{theorem}
Here $f\lesssim g$ means that there exists a constant $C>0$ such that $f\leq Cg$. In comparison we have achieved
\begin{equation}
m\gtrsim\frac{s}{\delta^2}\left((\ln s)^3\ln N+\ln\left(\frac1\epsilon\right)\right).\label{eq:optmas}
\end{equation}
In the sense that theorem \ref{cgv} is summarized in their paper, namely that the number of samples needed is of order $s(\ln s)^3\ln N$, we have not made any contribution (i.e. with regards to the asymptotics). However we think that for small $\epsilon$ the improvement is not insignificant. We do as well allow for a larger class of matrices and provide explicit constants. When constants have been presented before (for actually worse results in terms of asymptotics), as far we have seen they have been about a factor 10 larger than ours.\\
The main differences in the proofs lies in the arguments surrounding Dudley's inequality for Rademacher processes and that we do not make use of two different covering number estimates. The inequality requires a quite heavy proof, using probabilistic methods, c.f. \cite{Lihshits}. We re-use some of the arguments in that proof, but we first do pointwise estimates and then simply replace supremums with sums. One must take care when doing the covering and counting, details that we hope are perhaps a bit more clear through our exposition.


\section{Preliminaries}
\label{sec:prel}
We denote by $\|\cdot\|_p, 1\leq p<\infty$ the usual $\ell^p$ norm for vectors, $\|\mathbf{z}\|_0:=|\supp\mathbf{z}|$ denotes the cardinality of the support of a vector $\bz$ (sometimes called ''0-norm'', despite not being a norm) and $[N]=\{1,2,\dots,N\}$. In this work we will mostly restrict ourselves to vectors with real entries but one could easily generalize the results to complex vectors. 
By $\EV{X}$ we denote the expectation value with respect to a random variable, or random vector, $X$. In particular for the random sampling matrices with rows $\bX=\{X_j\}_{j=1}^m$ we will use $\EV{}$ to mean $\EV{\bX}=\EV{X_1}\EV{X_2}\cdots\EV{X_m}$ and otherwise be clear with subscripts if the expectation is taken in another random variable. Given a random variable $X$ and a measurable function $f$, we can for $1\leq p<\infty$ induce the $L^p$-norms $\|f\|_{X,p}=\EV{X}[|f(X)|^p]^{1/p}$.

\subsection{Sparsity and Restricted Isometry}
We start by defining what we mean by a sparse vector. In what follows, $N$ denotes a (usually large) positive integer.
\begin{definition}
\label{d:ssp}
$\mathbf{x}\in\C^N$ is called $s$-sparse if $\|\mathbf{x}\|_0\leq s.$
\end{definition}
The next definition will be of great use throughout this paper.
\begin{definition}
\label{d:sind}
If $\bx=(x_1,\dots,x_N), S\subset[N]$, we define $\bx_S=((x_{S})_1,\dots,(x_{S})_N)$ by $(x_S)_k=x_k\chi_S(k)$, where 
\[\chi_S(k)=\begin{cases}1, \textrm{ if }k\in S\\0, \textrm{ otherwise}\end{cases}\]
is the characteristic function of the set $S$. Clearly $\bx=\bx_S+\bx_{S^c}$, where $S^c=[N]\setminus S$.
\end{definition}
In practice one rather accepts small ''$s$-term approximation error'', i.e. one wants that the following quantity is small:
\[
\sigma_s(\mathbf{x})_p:=\inf\{\|\mathbf{x}-\mathbf{z}\|_p, \mathbf{z}\text{ is $s$-sparse}\}.
\] 
Think of $\mathbf{y}\in\C^m$ as the measured quantity from a measurement of $\mathbf{x}\in\C^N$, modelled after $\mathbf{y}=A\mathbf{x}$, where $A\in\C^{m\times N}$ is an $m\times N$-matrix and we assume that $m\ll N$.  In general this system is impossible to solve, unless we impose the extra condition that $\mathbf{x}$ is $s$-sparse and consider
\begin{equation}
\min_{\mathbf{z}\in\C^N}\|\mathbf{z}\|_0\quad\text{subject to}\quad A\mathbf{z}=\mathbf{y},\label{eq:l0min}
\end{equation}
in the hope that its solution $\mathbf{x}^*=\mathbf{x}$. This is still very hard to solve in general so one would like to consider the closest convex relaxation of \eqref{eq:l0min}, which is
\begin{equation}
\min_{\mathbf{z}\in\C^N}\|\mathbf{z}\|_1\quad\text{subject to}\quad A\mathbf{z}=\mathbf{y}.\label{eq:l1min}
\end{equation}
We ask when the solution of \eqref{eq:l1min} is equivalent to the solution of \eqref{eq:l0min}. The key notion is the so-called null space property for a matrix.
\begin{definition}
\label{d:nsp}
A matrix $A\in\C^{m\times N}$ satisfies the \emph{null space property} of order $s$ if for all subsets $S\subset[N]$ with $|S|=s$ it holds that
\begin{equation}
\|\mathbf{v}_S\|_1<\|\mathbf{v}_{S^c}\|_1\quad\text{for all }\mathbf{v}\in\ker A\setminus\{0\}.\label{eq:nsp}
\end{equation}
We write $A\in NSP(s)$.
\end{definition}
The following theorem gives the answer to when a solution of \eqref{eq:l0min} equals the solution of \eqref{eq:l1min}, for the proof see for example \cite{Rau1} (Theorem 2.3, p.8) or \cite{GrNi}.
\begin{theorem}
\label{Rau1:T2.3}
Let $A\in\C^{m\times N}$. Then every $s$-sparse vector $\mathbf{x}\in\C^N$ is the unique solution to the $\ell^1$-minimization problem \eqref{eq:l1min} with $\mathbf{y}=A\mathbf{x}$ if and only if $A$ satisfies the null space property of order $s$.
\end{theorem}
\noindent Below we present a helpful proposition that can be used to verify the null space property. The proof is a simple consequence of Lemma \ref{lem:l21} in the appendix where we sketch out the details. With a slightly more involved proof the propostion could be improved a bit further, replacing the constant $4/5$ with a constant arbitrarily close to (for large $s$)  $\sqrt{4/5}$. See further section \ref{sec:improvedelta}.
\begin{proposition}\label{prop:nsp2}
Assume $\mathbf{x}=(x_1,\dots,x_N)\in\C^N$ such that $|x_1|\geq|x_2|\geq\dots\geq|x_N|$. Write $\bx=\sum_k\bx_{S_k}$ where $S_1=\{1,\dots,s\},S_2=\{s+1,\dots,2s\}$ etc. so that $|S_k|=s$ (except for possibly the last $k$). Denote by $S^c=[N]\setminus S$. Then if 
\[
\|\bx_{S_1}\|_2<\frac45\sum_{k>1}\|\bx_{S_k}\|_2,
\]
it holds that $\|\bx_S\|_1<\|\bx_{S^c}\|_1$ for all subsets $S\subset[N]$ with $|S|=s$.
\end{proposition}
Unfortunately, the null space property is often hard to verify. Instead one usually tries to verify the weaker restricted isometry property for a matrix.
\begin{definition}
\label{d:rip}
The restricted isometry constants $\delta_s$ of a matrix $A\in\C^{m\times N}$ is defined as the smallest $\delta_s$ such that
\begin{equation}
(1-\delta_s)\|\mathbf{x}\|_2^2\leq\|A\mathbf{x}\|_2^2\leq(1+\delta_s)\|\mathbf{x}\|_2^2 \label{eq:rip}
\end{equation}
for all $s$-sparse $\mathbf{x}\in\C^N$. We abbreviate this by $A\in RIP(\delta_s)$.
\end{definition}
Another characterization of the restricted isometry constants is given by:
\begin{proposition}[\cite{Rau1}:2.5 (p.9)]
\label{Rau1:P2.5}
Let $A\in\C^{m\times N}$, with restricted isometry constants $\delta_s$, then
\[
\delta_s=\sup_{\mathbf{x}\in T_s}|\langle(A^*A-I)\mathbf{x},\mathbf{x}\rangle|,
\textrm{ where } \ T_s=\{\mathbf{x}\in\C^N, \|\mathbf{x}\|_2=1,\|\mathbf{x}\|_0\leq s\}.\]
\end{proposition}
The restricted isometry property can, under some extra condition, imply the null space property as the following theorem suggests.
\begin{theorem}
\label{thm:best_delta}
Suppose the restricted isometry constants $\delta_{2s}$ of a matrix $A\in\C^{m\times N}$ satisfies
\[\delta_{2s}<\frac4{\sqrt{41}}\approx0.62,\]
then the null space property of order $s$ is satisfied. In particular, every $s$-sparse vector $\mathbf{x}\in\C^N$ is recovered by $\ell^1$-minimization.
\end{theorem}
This is an improvement of the best known result, from \cite{MLric}, which had $\delta_{2s}<0.4931$ (see also \cite{Fou1},\cite{CWX1},\cite{CWX2}). The proof will be included in the appendix. With some more work the authors can replace the constant $4/\sqrt{41}$ with a constant, arbitrarily close to for large $s$, 2/3. The key ingredient is the mentioned improvement of proposition \ref{prop:nsp2}. See further section \ref{sec:improvedelta}. The best we can hope for is to replace the constant with $1/\sqrt 2$, due to the work in \cite{DaGr}. \\

\subsection{Entropy and Low Entropy Isometry}
Next we will define the \emph{$\ell^1$-entropy} (also known as the \emph{$\ell^1$-sparsity level}, as defined in for example \cite{Tang}) which is closely related to sparseness.

\begin{definition}
By the $\ell^1$-entropy of a nonzero vector $\bx\in\R^n$ we mean the quantity
\[
\Ent(\bx)=\frac{\|\bx\|_1^2}{\|\bx\|_2^2}.
\]
\end{definition}

\begin{remark}
Clearly if $\bx$ is $s$-sparse then $\Ent(\bx)\leq s$ by Cauchy-Schwarz inequality.
\end{remark}
In replacement of null space property,  one has the null entropy property.
\begin{definition}
A matrix $A\in\C^{m\times N}$ satisfies the \emph{null entropy property of order $t$} if for every $\bx\in\ker A\setminus\{\mathbf{0}\}$ it holds that $\Ent(\bx)\geq t$. We write $A\in NEP(t)$.
\end{definition}
A low entropy isometry property can be defined as well, analogous with the restricted isometry property.
\begin{definition}
A matrix $A\in\C^{m\times N}$ satisfies the \emph{low entropy isometry property} with constants $\tilde\delta_t$ if for all $\bx$ with $\Ent(\bx)\leq t$, 
\[
|\|A\bx\|_2^2-\|\bx\|_2^2|\leq\tilde\delta_t\|\bx\|_2^2.
\]
We abbreviate this by $A\in LEIP(\tilde\delta_t)$.
\end{definition}

Many of the above notions are related by the following proposition:
\begin{proposition}\
\label{prop1}
\begin{enumerate}
\item If $t>4s$ and $A\in NEP(t)$ then $A\in NSP(s)$.
\item If $\tilde\delta_t<1$ and $A\in LEIP(\tilde\delta_t)$, then $A\in NEP(t)$.
\item If $s\leq t$ and $A\in LEIP(\tilde\delta_t)$, then $A\in RIP(\delta_s)$ for some $\delta_s\leq\tilde\delta_t$.
\end{enumerate}
\end{proposition}

A variant of $\it1$ can be found in \cite{Tang}, and both that and $\it2$ can be proved on a single line by considering the contrapositive statements while $\it3$ is obvious.

\subsection{Bounded orthonormal systems}

Let $\mathcal{D}\subset\R^d$, $\nu$ a probability measure on $\mathcal{D}$, $\{\psi_j\}_{j=1}^N$ a bounded orthonormal system of complex-valued functions on $\mathcal{D}$. This means that for $j,k\in[N]$,
\begin{equation}
\label{eq:ON}
\int_\mathcal{D}\psi_j(t)\overline{\psi_k(t)}d\nu(t)=\delta_{jk},
\end{equation}
and $\{\psi_j\}$ is uniformly bounded in $L^\infty$, 
\begin{equation}
\label{eq:ubON}
\|\psi_j\|_\infty=\sup_\mathcal{D}|\psi_j(t)|\leq K\quad\textrm{for all }j\in[N], (K\geq1).
\end{equation}
Let now $t_1\dots t_m\in\D$ (picked independently at random with respect to $\nu$) and suppose we are given sample values
\[
y_l=f(t_l)=\sum_{k=1}^Nx_k\psi_j(t_l),\quad l=1,\dots,m.
\]
Introduce $A\in\C^{m\times N}, A = (a_{lk}), a_{lk}=\psi_k(t_l), l = 1,\dots,m; k=1,\dots,N.$ Then $\mathbf{y}=A\mathbf{x}, \mathbf{y}=(y_1,\dots,y_m)^T$ and $\mathbf{x}$ is a vector of coefficients. We wish to reconstruct the polynomial $f$ (or equivalently $\mathbf{x}$) from the samples $\mathbf{y}$, using as few samples as possible. If we assume that $f$ is $s$-sparse (defined to be so if $\mathbf{x}$ is $s$-sparse) the problem reduces to solving $\mathbf{y}=A\mathbf{x}$ with a sparsity constraint. $P(t_l\in B)=\nu(B)$ for measurable $B\subset\D$, so $A$ becomes a random sampling matrix (fulfills \eqref{eq:ON},\eqref{eq:ubON} and $t_l$ are picked independently at random with respect to $\nu$). One interesting example is given by sampling $m$ rows from the $N\times N$-matrix
\[
a_{lk}=\frac{\e^{2\pi\i lk/N}}{\sqrt N}, \ l,k\in[N].
\]
This matrix is called a \emph{random partial Fourier matrix}. We summarize this section with a definition of the matrices we will continue to study.
\begin{definition}[Random Sampling Matrix]
A matrix $A\in\C^{m\times N}$ is said to be a \emph{random sampling matrix} if its rows $\bX=\{X_j\}_{j=1}^m$ fulfills the conditions:
\begin{enumerate}
\item $\|X_j\|_\infty\leq K$ for some $K\geq1$.
\item $\EV{}[X_j^*X_j]=I_N$ ($N\times N$ identity matrix), for all $j$.
\end{enumerate}
\end{definition}

\section{Preparatory lemmas and inequalities}\label{sec:lem}
We move on to present some key ingredients to be used in the proof of the main theorem of this paper. First we remind about the definition of a Rademacher sequence.
\begin{definition}
A \emph{Rademacher sequence} $\boldsymbol\eps=(\eps_j)_{j=1}^m$ is a random vector whose components $\eps_j$ takes the values $\pm1$ with equal probability ($=\frac12$).
\end{definition}
Symmetrization is a useful technique that will later be used to bound the expectation value of the restricted isometry constants $\delta_s$. The proof of the proposition is not very hard and can be found in for example \cite{LTProb} or \cite{Rau1}.
\begin{proposition}[Symmetrization]
\label{prop:sym}
Assume that $\boldsymbol\xi=(\xi_j)_{j=1}^m$ is a sequence of independent random vectors in $\C^N$ equipped with a (semi-) norm $\|\cdot\|$, having expectations $x_j=\EV{}\xi_j$. Then for $1\leq p<\infty$
\[
\left(\EV{}\|\sum_{j=1}^m(\xi_j-x_j)\|^p\right)^{1/p}\leq2\left(\EV{}\|\sum_{j=1}^m\eps_j\xi_j\|^p\right)^{1/p}
\] 
where $\boldsymbol\eps=(\eps_j)_{j=1}^m$ is a Rademacher sequence independent of $\boldsymbol\xi$.
\end{proposition}
Khintchine's inequality is another important inequality to be used later on.
\begin{proposition}[Khintchine's inequality]
\label{prop:khin}
Suppose $\bx=(x_1,\dots,x_N)\in\C^N$ and $\boldsymbol\eps=(\eps_1,\dots,\eps_N)$ is a vector whose components are independent Rademacher random variables, then for $p\geq2$
\begin{equation}\label{eq:khin2}
\EV{\boldsymbol\eps}\left|\sum_{j=1}^N\eps_jx_j\right|^{p}\leq2^{3/4}\left(\frac{p}{\e}\right)^{p/2}\|\bx\|_2^{p}.
\end{equation}
\end{proposition}
The proof can be found in a lot of literature, see for example \cite{Rau1}, p.35.

\subsection{Covering and packing estimates}

We will work in the framework of a random sampling matrix (with rows $\bX=\{X_j\}_{j=1}^m$, $\|X_j\|_\infty\leq K$) and introduce the metric
\[
d_{\bX,p}(\bx,\by)=\left(\frac1m\sum_{j=1}^m|\langle X_j,\bx-\by\rangle|^p\right)^{1/p}.
\]
$B_{\bX,p}(\bx,r)=\{\by\in\R^N:d_{\bX,p}(\bx,\by)<r\}$ denotes the ball of radius $r>0$ around $\bx\in\R^N$ with respect to the metric $d_{\bX,p}$. The next lemma is based on the method of Maurey.

\begin{lemma}[Covering lemma 1]
\label{lem:c1}
Let $0<r<K$, $p\geq1$,
\begin{equation}
M\geq2^{\frac3{4p}}\frac{8pK^2}{r^2\e}\label{eq:M-r}
\end{equation}
and let $G_M=\{\bz_j\}$ be the set of grid points in the $\ell^1$ unit cube with mesh size $\frac1M$, i.e. the set of points satisfying $\|\bz\|_1\leq1$ and $M\bz\in\Z^N$. Then $B_1=\{\bz\in\R^N;\|\bz\|_1\leq1\}$ is contained in $\cup_j B_{\bX,2p}(\bz_j,r)$ for some fix realization of $\bX=\{X_j\}$, with the property $\|X_j\|_\infty< K$ and $r$ given by equality in \eqref{eq:M-r}. The number of grid points is less than
\[
{{2N+M}\choose{M}}\leq\left(\frac{2N\e}{M}+\e\right)^M.
\]
\end{lemma}
\begin{proof}[Proof of lemma \ref{lem:c1}]
Fix a point in $\bx=(x_1,\dots,x_N)\in B_1$ and define a random vector $Z=(z_1,\dots,z_N)$ by letting it take the value $\sgn(x_j)\be_j$ with probability $|x_j|$, and $Z=\bold{0}$ with probability $1-\|\bx\|_1$ (so $\|Z\|_0\leq1$). Let now $Z_k, k=1,\dots,M$ be $M$ independent copies of $Z$ and define
\[
\bz=\frac1M\sum_{k=1}^M Z_k.
\] 
Then $\bz\in G_M$ and $\EV{Z}\bz=\bx$. Now it is enough to prove that
\[
\frac1m\EV{Z}\sum_{j=1}^m|\langle X_j,\bz-\bx\rangle|^{2p}<r^{2p}
\]
for some $p\geq1$. By symmetrization and Khintchine's inequality applied to every term,
\begin{multline*}
\frac1m\sum_{j=1}^m\EV{Z}|\langle X_j,\bz-\bx\rangle|^{2p}\leq\frac1m\sum_{j=1}^m2^{2p}\EV{Z}\EV{\boldsymbol{\eps}}\left|\frac1M\sum_{k=1}^M\eps_k|\langle X_j,Z_k\rangle|\right|^{2p}\leq\\
\frac1m\sum_{j=1}^m\left(\frac2M\right)^{2p}2^{3/4}\left(\frac{2p}\e\right)^p\EV{Z}\left(\sum_{k=1}^M|\langle X_j,Z_k\rangle|^2\right)^p<2^{3/4}\left(\frac{8p}{M\e}\right)^pK^{2p}=:r^{2p}.
\end{multline*}
The number of balls needed for the cover follows from simple combinatorics. We can choose $M$ vectors out of the collection $\{\pm\be_j\}_{j=1}^N\cup\{\boldsymbol{0}\}$ in less than ${{2N+1+M-1}\choose M}$ ways (i.e we count the number of unordered selections with repetition allowed). It is also well-known that
\begin{align*}
{{2N+M}\choose M}\leq\left(\frac{2N\e}{M}+\e\right)^M.
\end{align*}
\end{proof}

\begin{remark}
We will use Lemma \ref{lem:c1} for $\bz\in B_1(0,\sqrt{s}), M=2^{2k}$, so the radii of the balls in the cover will then be 
\[
r_k=2^{-k}2^{\frac1{4p}}K\left(\frac{8ps}{\e}\right)^{1/2},
\]
and the number of balls in the cover (the covering number) for this $k$ will be
\[
N_k=\left(\frac{2N\e}{2^{2k}}+\e\right)^{2^{2k}}.
\]
\end{remark}

\section{Uniform recovery theorem}
\label{sec:mains}

The following technical lemma is going to be the key ingredient and we postpone the rather involved proof until the end of this section.

\begin{lemma}
\label{lem:cru}
Let $A\in\C^{m\times N}$ be a random sampling matrix with corresponding low entropy isometry (or restricted isometry) constants $\delta_s$ and rows $\{X_j\}_{j=1}^m$ having the properties that $\|X_j\|_\infty<K$ for some $K\geq1$ and $\EV{}[X_j^*X_j]=I_N$ for all $j$. Suppose that $N>4p, p=\ln(2^{3/4}K^2s)\geq2, 0<\lambda,g<1$, then
\begin{equation*}
(\EV{}\delta_s^{2n})^{\frac1{2n}}\leq(H+\lambda g)((\EV{}\delta_s^{2n})^{\frac1{2n}}+1)^{\frac1{2q}}
\end{equation*}
where $q=q(K,s)\in(1,2]$ and
\begin{multline*}
H=H(N,K,m,s,\lambda,g)=\\\left(\frac{2^{10}K^2s}{(\ln2)^2m}\right)^{1/2}\left(p^{1/2}\ln\left(\frac{2^6\e K^2s}{(\lambda g)^2}\right)\ln^{1/2}(N/p)+\ln^{1/2}(1/\epsilon^\alpha)\right),\ \alpha=\frac{\e\ln2}{8}.
\end{multline*}
\end{lemma}

Using lemma \ref{lem:cru} we can prove:
\begin{theorem}
\label{thm:main1}
Let $A\in\C^{m\times N}$ be a random sampling matrix with corresponding low entropy isometry (or restricted isometry) constants $\delta_s$ and rows $\{X_j\}_{j=1}^m$ having the properties that $\|X_j\|_\infty<K$ for some $K\geq1$ and $\EV{}[X_j^*X_j]=I_N$ for all $j$. Suppose $0<\delta,\epsilon,\lambda<1$ and
\begin{equation}
\sqrt{m}>C_1K\sqrt{s}\left(\ln^{1/2}(2^{3/4}K^2s)\ln(C_2 K^2s)\ln^{1/2}(N)+
\ln^{1/2}(1/\epsilon)\right) \label{m_est11}
\end{equation}
where 
\begin{align*}
C_1(\delta,\lambda)=\frac{2^{5}\e^{1/4}}{\ln2}\frac{(\sqrt\e+\delta)^{1/2}}{(1-\lambda)\delta},\quad C_2(\delta,\lambda)=\frac{2^6\e^{3/2}(\delta+\sqrt\e)}{(\delta\lambda)^2}\ 
\end{align*}
Then $P(\delta_s>\delta)<\epsilon$, that is $\frac1{\sqrt{m}}A$ has the low entropy (or restricted) isometry property with constants $\delta_s\leq\delta$ with probability $1-\epsilon$.
\end{theorem}

\begin{proof}
Since in our framework $s\leq m\ll N$, by Markov's inequality, for any $n>0$,
\[
P(\delta_s>\delta)=P(\delta_s^{2n}>\delta^{2n})\leq\frac{\EV{}\delta_s^{2n}}{\delta^{2n}}<\epsilon.
\]
By lemma \ref{lem:cru}, this is less than $\epsilon\in(0,1)$ if,
\begin{equation}
H+\lambda g\leq\frac{\delta\epsilon^{\frac1{2n}}}{(\delta\epsilon^{\frac1{2n}}+1)^{\frac1{2q}}}.\label{eq:mark}
\end{equation}
Choosing $n\geq\ln\left(\frac1\epsilon\right)$ implies that $\epsilon^{\frac1{2n}}\geq\frac1{\sqrt\e}$, and with this choice \eqref{eq:mark} is easily seen to be implied by
\begin{equation}
H+\lambda g<\frac{\delta}{(\delta\sqrt\e+\e)^{\frac1{2}}}.\label{eq:mark2}
\end{equation}
Define the right hand side expression to be $g=g(\delta)$, then
\begin{multline*}
H<(1-\lambda)g(\delta)\iff\\
\left(\frac{2^{10}K^2s}{(\ln2)^2m}\right)^{1/2}\left(p^{1/2}\ln\left(\frac{2^6\e K^2s}{(\lambda g)^2}\right)\ln^{1/2}(N/p)+\ln^{1/2}(1/\epsilon^\alpha)\right)<(1-\lambda)g\iff\\
\sqrt{m}>\frac{2^{5}\e^{1/4}}{\ln2}\frac{(\sqrt\e+\delta)^{1/2}}{(1-\lambda)\delta}K\sqrt{s}\left(p^{1/2}\ln\left(\frac{2^6\e K^2s}{(\lambda g)^2}\right)\ln^{1/2}(N/p)+\ln^{1/2}(1/\epsilon^\alpha)\right).
\end{multline*}
Since $\alpha<1$, and by removing some lower order terms, \eqref{m_est11} can be seen to imply \eqref{eq:mark2}, so we are done.
\end{proof}
\begin{remark}\label{rmk:dev}
We could modify the proof, choosing $n$ larger so that $\epsilon^{1/2n}$ comes arbitrarily close to $1$, compared to above where we only used $\e^{-1/2}$ as lower bound. This corresponds to constants we would get by doing an argument closer to what is done for the best result in for example \cite{Rau1}, where the so-called deviation inequality is used.
\end{remark}
If we introduce
\[
C(\delta,\lambda)=\sqrt{2}C_1(\delta,\lambda), \ D(\delta,\lambda)=2C_2(\delta,\lambda)
\]
we get together with theorem \ref{thm:best_delta} the following corollary to theorem \ref{thm:main1}:

\begin{corollary}
\label{cor:1}
Let $A\in\C^{m\times N}$ be a random sampling matrix with corresponding restricted isometry constants $\delta_s$ and rows $\{X_j\}_{j=1}^m$ having the properties that $\|X_j\|_\infty<K$ for some $K\geq1$ and $\EV{}[X_j^*X_j]=I_N$ for all $j$. Suppose $0<\epsilon,\lambda<1$ and if
\begin{multline}
\sqrt{m}>C\left(\frac{4}{\sqrt{41}},\lambda\right)K\sqrt{s}\ln^{1/2}(2^{3/4}K^2s)\ln\left(D\left(\frac{4}{\sqrt{41}},\lambda\right) K^2s\right)\ln^{1/2}(N/p)+\\
C\left(\frac{4}{\sqrt{41}},\lambda\right)K\sqrt{s}\ln^{1/2}(1/\epsilon) \label{m_est_nsp}
\end{multline}
then with probability $1-\epsilon$, the matrix $\frac1{\sqrt{m}}A$ satisfies the null space property of order $s$.
\end{corollary}
\begin{remark}
Another variant of the above corollary would be to instead demand that the low entropy isometry constants $\tilde\delta_{4s}<1$ and use proposition \ref{prop1}.
\end{remark}
Below we present tables of values of $C^2$ (for convenience these are easier to compare with older results) and $D$ for some interesting choices of $\delta$ and $\lambda$.

\begin{table}[h!]
\caption{Some values of $C(\delta,\lambda)^2, D(\delta,\lambda)$.}
\begin{center}
\begin{tabular}{|c|c|c|c|c|}
\hline
$\lambda$ & $\left\lceil C\left(\frac4{\sqrt{41}},\lambda\right)^2\right\rceil$ & $\left\lceil D\left(\frac4{\sqrt{41}},\lambda\right)\right\rceil$ & $\left\lceil C\left(\frac23,\lambda\right)^2\right\rceil$ & $\left\lceil D\left(\frac23,\lambda\right)\right\rceil$ \\\hline
0&40943&$\infty$&36613&$\infty$\\\hline
1/9&51818&270695&46339&242072\\\hline
1/2&163769&13368&146452&11955\\\hline
$1/\sqrt{\e}$&264453&9085&236489&8124\\\hline
1&$\infty$&3342&$\infty$&2989\\\hline
\end{tabular}
\end{center}
\end{table}

\begin{remark}
Note however that squaring for example \eqref{m_est_nsp} in order to arrive at an expression such as \eqref{eq:optmas}, $C^2$ needs to be multiplied with something like $1+\beta$ (using for example Young's inequality), but $\beta>0$ can be chosen very small.
\end{remark}
Asymptotically, in the sense of remark \ref{rmk:dev}, we could gain about a factor $\e$. So optimal lower bounds using our methods are given by:
\begin{align*}
17747\leq\left\lceil C\left(\frac4{\sqrt{41}},0\right)^2\right\rceil,\quad& 1449\leq\left\lceil D\left(\frac4{\sqrt{41}},1\right)\right\rceil\\
15985\leq\left\lceil C\left(\frac23,0\right)^2\right\rceil,\quad& 1305\leq\left\lceil D\left(\frac23,1\right)\right\rceil
\end{align*}

\begin{proof}[Proof of lemma \ref{lem:cru}]
First note that
\begin{align*}
\EV{X}\frac1m\sum_{j=1}^m|\langle X_j,\bu\rangle|^2=\frac1m\sum_{j=1}^m\bu\EV{X_j}[X_j^*X_j]\bu^*=\frac1mm\langle \bu,\bu\rangle=\|\bu\|_2^2.
\end{align*}
We will do the proof for the low entropy isometry constants, then the same conclusion will hold for the restricted isometry constants since they are always smaller. Let $\U=\{\bu\in\R^N;\|\bu\|_1\leq\sqrt{s},\|\bu\|_2\leq1\}$, by the symmetrization inequality (prop. \ref{prop:sym}), Fatou's lemma and the definition of $\delta_s$ (as in proposition \ref{Rau1:P2.5}(b), a similar definition holds for the low entropy isometry constants when we take supremum over the larger set $\U$), we get
\begin{align*}
\EV{}\delta_s^{2n}=\EV{}\sup_{\bu\in\mathcal{U}}\left|\frac1m\sum_{j=1}^m|\langle X_j,\bu\rangle|^2-\|\bu\|_2^2\right|^{2n}\leq
2^{2n}\EV{}\EV{\boldsymbol{\eps}}\sup_{\bu\in\mathcal{U}}\left|\frac1m\sum_{j=1}^m\eps_j|\langle X_j,\bu\rangle|^2\right|^{2n}
\end{align*}
where $\boldsymbol{\eps}=\{\eps_j\}_{j=1}^m$ is a Rademacher sequence.
Let us now fix a realization of the $X_j=:\bx_j$ and define
\[
E_{2n}:=\left(\EV{\eps}\sup_{\bu\in\mathcal{U}}\left|\frac1m\sum_{j=1}^m\eps_j|\langle \bx_j,\bu\rangle|^2\right|^{2n}\right)^{1/2n}, \textrm{ so }\EV{}\delta_s^{2n}\leq\EV{}[(2E_{2n}(X))^{2n}].
\]
 
By lemma \ref{lem:c1}, for every $\bu\in\U$ there exists a gridpoint $\bz_k\in G_k:=(2^{-2k}\sqrt{s})\Z^N\cap B_1(0,\sqrt{s})$, (where $B_1(0,\sqrt{s})=\{\bz\in\R^N:\|\bz\|_1<\sqrt{s}\}$ and since $\U\subset\{\bu\in\R^N;\|\bu\|_1\leq\sqrt{s}, \|\bu\|_2\leq1\}$) such that for any $p\geq1,$
\[
d_{\bX,2p}(\bu,\bz_k)<r_k(p).
\]
For every $\bz_k\in G_k$ consider
\[
B_{\bX,2p}(\bz_k,r_k)=\{\bz\in\R^N:d_{\bX,2p}(\bz,\bz_k)<r_k(p)\}.
\] 
If $\mathcal{U}\cap B_{\bX,2p}(\bz_k,r_k)\neq\emptyset$, pick an arbitrary element from this set and denote it $\pi_k \bu$, then we get a finite cover of $\U$ with balls $B_x(\pi_k\bu,2r_k)$. We will do this for $l\leq k\leq L$ where $l$ and $L$ are to be determined. Denote by $\U_k:=\{\pi_k\bu:\bu\in \U_{k+1}\}$ and note that  $|\U_k|\leq|G_k|\leq N_k<\infty$. Now we get using telescoping sums, and the conventions $\U_{L+1}=\U,\Pi_{L+1}\bu=\bu$
\begin{multline*}
\sum_{j=1}^m\eps_j|\langle \bx_j,\bu\rangle|^2=\sum_{j=1}^m\sum_{k=l+1}^{L+1}\eps_j(|\langle \bx_j,\Pi_k\bu\rangle|^2-|\langle \bx_j,\Pi_{k-1}\bu\rangle|^2)+\sum_{j=1}^m\eps_j|\langle \bx_j,\Pi_{l}\bu\rangle|^2\\
\implies
\left|\frac1m\sum_{j=1}^m\eps_j|\langle \bx_j,\bu\rangle|^2\right|\leq
\left|\frac1m\sum_{j=1}^m\eps_j(|\langle \bx_j,\bu\rangle|^2-|\langle \bx_j,\Pi_{L}\bu\rangle|^2)\right|+\\
\sum_{k=l+1}^{L}\left|\frac1m\sum_{j=1}^m\eps_j(|\langle \bx_j,\Pi_k\bu\rangle|^2-|\langle \bx_j,\Pi_{k-1}\bu\rangle|^2)\right|
+\left|\frac1m\sum_{j=1}^m\eps_j|\langle \bx_j,\Pi_l\bu\rangle|^2\right|,
\end{multline*}
where $\Pi_k\bu=\pi_k\circ\pi_{k+1}\circ\cdots\circ\pi_L\bu$. Then we get
\begin{multline*}
E_{2n}=\left(\EV{\boldsymbol{\eps}}\sup_{\bu\in\mathcal{U}}\left|\frac1m\sum_{j=1}^m\eps_j|\langle \bx_j,\bu\rangle|^2\right|^{2n}\right)^{1/2n}\leq\\
\left(\EV{\boldsymbol{\eps}}\sup_{\bu\in\mathcal{U}}\left|\frac1m\sum_{j=1}^m\eps_j(|\langle \bx_j,\bu\rangle|^2-|\langle \bx_j,\pi_{L}\bu\rangle|^2)\right|^{2n}\right)^{1/2n}+\\
\left(\EV{\boldsymbol{\eps}}\left[\sum_{k=l+1}^{L}\sup_{\bu\in\mathcal{U}_k}\left|\frac1m\sum_{j=1}^m\eps_j(|\langle \bx_j,\bu\rangle|^2-|\langle \bx_j,\pi_{k-1}\bu\rangle|^2)\right|\right]^{2n}\right)^{1/2n}+\\
\left(\EV{\boldsymbol{\eps}}\sup_{\Pi_l\bu\in\U_l}\left|\frac1m\sum_{j=1}^m\eps_j|\langle \bx_j,\Pi_l\bu\rangle|^2\right|^{2n}\right)^{1/2n}=:S_{L+1}+S_{l+1,L}+S_l.
\end{multline*}

In order to estimate $S_{l+1,L}$ we introduce 
\begin{align*}
g_k(\boldsymbol{\eps},\bu):=\left|\frac1m\sum_{j=1}^m\eps_j(|\langle \bx_j,\bu\rangle|^2-|\langle \bx_j,\pi_{k-1}\bu\rangle|^2)\right|,\quad \bu\in\U_k\textrm{ and}\\f_k(\boldsymbol{\eps}):=\sup_{\bu\in\mathcal{U}_k}g_k(\boldsymbol{\eps},\bu).
\end{align*}
We also specify norm notations using
\[
\|f\|_{\boldsymbol{\eps},2n}:=(\EV{\boldsymbol{\eps}}|f|^{2n})^{1/2n},\textrm{ we can write } S_{l+1,L}=\left\|\sum_{k=l+1}^{L}f_k\right\|_{\boldsymbol{\eps},2n}
\]
We will derive auxiliary estimates for $S_{l}, \|f_k\|_{\boldsymbol{\eps},2n}$ and $S_{L+1}$, summarized in
\begin{lemma}\label{lem:auxest}
For any non-negative integers $l\leq k\leq L$, there are $p>q>1$ (depending on $K$ and $s$), $\frac1p+\frac1q=1$, such that for any positive integer $n$ the following estimates hold:
\begin{eqnarray}
S_l&\leq&\left(\frac{2K^2sn}{m}\right)^{1/2}(2^{3/4}N_l)^{1/2n}S^{1/q}\label{est:S_l}\\
\|f_k\|_{\boldsymbol{\eps},2n}&\leq&\left(\frac{2^{10-2k}K^2snp}{m}\right)^{1/2}\left(\frac{(2^{3/4}N_k)^{1/n}}{\e}\right)^{1/2}S^{1/q}\label{est:f_k}\\
S_{L+1}&\leq&(2^{7-2L}K^2sp)^{1/2}S^{1/q}\label{est:S_L+1}
\end{eqnarray}
where $N_k\geq|\U_k|$ and
\[
S=S(\bx)=\sup_{\bu\in\U}\left(\frac1m\sum_{j=1}^m|\langle \bx_j,\bu\rangle|^2\right)^{1/2}.
\]
\end{lemma}

\begin{proof}[Proof of lemma \ref{lem:auxest}]
There are many similarities in proving the above estimates. If we first consider $S_l^{2n}$ for a fixed $\Pi_l\bu$ it follows by Khintchine's and Hölder's inequalities, that
\begin{multline*}
\EV{\boldsymbol{\eps}}\left|\frac1m\sum_{j=1}^m\eps_j|\langle \bx_j,\Pi_l\bu\rangle|^2\right|^{2n}\leq
2^{3/4}\left(\frac{2n}{m\e}\right)^n\left(\frac1m\sum_{j=1}^m|\langle \bx_j,\Pi_l\bu\rangle|^4\right)^n\leq\\
2^{3/4}\left(\frac{2n}{m\e}\right)^n\left(\frac1m\sum_{j=1}^m|\langle \bx_j,\Pi_l\bu\rangle|^{2p}\right)^{n/p}\left(\frac1m\sum_{j=1}^m|\langle \bx_j,\Pi_l\bu\rangle|^{2q}\right)^{n/q}\leq\\
2^{3/4}\left(\frac{2n}{m\e}\right)^n(K^2s)^n\left(\frac1m\sum_{j=1}^m|\langle \bx_j,\Pi_l\bu\rangle|^{2+2q/p}\right)^{n/q}\leq\\
2^{3/4}\left(\frac{2K^2sn}{m\e}\right)^n(K^2s)^{n/p}\left(\sup_{\bu\in\U}\frac1m\sum_{j=1}^m|\langle \bx_j,\bu\rangle|^2\right)^{n/q}=\\
2^{3/4}\left(\frac{2K^2sn}{m\e}\right)^n(K^2s)^{n/p}S^{2n/q}.
\end{multline*}
After the second two lines we simply used that $\|\bx_j\|_\infty\leq K$ and $\|\Pi_l\bu\|_1\leq\sqrt s$ and thus $|\langle\bx_j,\Pi_l\bu\rangle|^2\leq K^2s$. Since the derived estimate holds for any $\Pi_l\bu\in\U_l$ we can use the trivial inequality
\[
\EV{\boldsymbol{\eps}}\sup_{\bu\in\U_k}|f(\boldsymbol{\eps},\bu)|\leq\EV{\boldsymbol{\eps}}\sum_{\bu\in\U_k}|f(\boldsymbol{\eps},\bu)|\leq N_kA
\]
which holds whenever $\EV{\boldsymbol{\eps}}|f(\boldsymbol{\eps},\bu)|\leq A$ and $|\U_k|\leq N_k$ to get
\[
S_l^{2n}\leq N_l\cdot2^{3/4}\left(\frac{2K^2sn}{m\e}\right)^n(K^2s)^{n/p}S^{2n/q}.
\]
In the proof of \eqref{est:f_k}, we will choose $p$ large enough to ensure $(K^2s)^{1/p}\leq\e$. Taking this into account, combined with taking the $2n$:th root of the above inequality, shows \eqref{est:S_l}:
\[
S_l\leq\left(\frac{2K^2sn}{m}\right)^{1/2}(2^{3/4}N_l)^{1/2n}S^{1/q}.
\]
In the same manner one shows for fixed $\bu\in\U_k$,
\begin{multline*}
\EV{\boldsymbol{\eps}}g_k(\boldsymbol{\eps},\bu)^{2n}\leq\\
2^{3/4}\left(\frac{2n}{m\e}\right)^n\left(\frac1m\sum_{j=1}^m(|\langle \bx_j,\bu\rangle|-|\langle \bx_j,\pi_{k-1}\bu\rangle|)^{2p}\right)^{n/p}\cdot\\
\left(\frac1m\sum_{j=1}^m(|\langle \bx_j,\bu\rangle|+|\langle \bx_j,\pi_{k-1}\bu\rangle|)^{2q}\right)^{n/q}\leq\\
2^{3/4}\left(\frac{2n}{m\e}\right)^nd_{\bX,2p}(\bu,\pi_{k-1}\bu)^{2n}\left(\sup_{\bu\in\U}\frac1m\sum_{j=1}^m(2|\langle \bx_j,\bu\rangle|)^{2q}\right)^{n/q}\leq\\
2^{3/4}\left(\frac{2n}{m\e}\right)^n(2r_{k-1}(p))^{2n}4^n(K^2s)^{n/p}S^{2n/q})=\\
2^{3/4}\left(\frac{2^{10-2k}K^2spn}{m\e^2}\right)^n(2^{3/4}K^2s)^{n/p}S^{2n/q}
\end{multline*}
where we plugged in $r_{k-1}(p)=2^{1-k}2^{\frac3{8p}}K\left(\frac{8ps}{\e}\right)^{1/2}$ from the remark following lemma \ref{lem:c1}. Since the above is valid for all $\bu\in\U_k$, we get (similarly as for $S_l$)
\[
\|f_k\|_{\boldsymbol{\eps},2n}\leq\left(\frac{2^{10-2k}K^2spn}{m\e^2}\right)^{1/2}(2^{3/4}K^2s)^{1/2p}(2^{3/4}N_k)^{1/2n}S^{1/q},
\]
where $N_k$ are also chosen as in the remark following lemma \ref{lem:c1}. Choosing $p=\ln(2^{3/4}K^2s)$, ensures that $(2^{3/4}K^2s)^{1/2p}=\e^{1/2}$ which concludes the proof of \eqref{est:f_k}.\\
Lastly, fixing $\bu\in\U$, using Cauchy-Schwarz and Hölder's inequalities,
\begin{multline*}
\EV{\boldsymbol{\eps}}\left|\frac1m\sum_{j=1}^m\eps_j(|\langle \bx_j,\bu\rangle|^2-|\langle \bx_j,\pi_{L}\bu\rangle|^2)\right|^{2n}\leq\\
\frac1{m^{2n}}\EV{\boldsymbol{\eps}}\left|\left(\sum_{j=1}^m(|\langle \bx_j,\bu\rangle|^2-|\langle \bx_j,\pi_{L}\bu\rangle|^2)^2\right)^{1/2}\left(\sum_{j=1}^m\eps_j^2\right)^{1/2}\right|^{2n}=\\
\left(\frac1m\sum_{j=1}^m(|\langle \bx_j,\bu\rangle|-|\langle \bx_j,\pi_{L}\bu\rangle|)^2(|\langle \bx_j,\bu\rangle|+|\langle \bx_j,\pi_{L}\bu\rangle|)^2\right)^n\leq\\
\left(\frac1m\sum_{j=1}^m(|\langle \bx_j,\bu\rangle|-|\langle \bx_j,\pi_{L}\bu\rangle|)^{2p}\right)^{n/p}\left(\frac1m\sum_{j=1}^m(|\langle \bx_j,\bu\rangle|+|\langle \bx_j,\pi_{L}\bu\rangle|)^{2q}\right)^{n/q}\leq\\
(4r_L(p))^{2n}(K^2s)^{n/p}S^{2n/q}=\left(\frac{2^{7-2L}K^2sp}{\e}\right)^{n}(2^{3/4}K^2s)^{n/p}S^{2n/q}=\\
(2^{7-2L}K^2sp)^{n}S^{2n/q}.
\end{multline*}
Since this holds for any $\bu\in\U$, \eqref{est:S_L+1} follows by taking a $2n$:th root.
\end{proof}
Comparing the bounds in \eqref{est:S_l} and \eqref{est:f_k} for $k=l$, one easily sees that choosing 
\[
l:=\left\lfloor\frac12\log_2\left(\frac{2^9p}{\e}\right)\right\rfloor\leq\frac12\log_2\left(\frac{2^9p}{\e}\right)
\]
implies that
\[
S_l\leq\left(\frac{2^{10-2l}K^2snp}{m}\right)^{1/2}\left(\frac{(2^{3/4}N_l)^{1/n}}{\e}\right)^{1/2}S^{1/q}.
\]
Next we will define an increasing sequence $\{n_k\}_{k=l}^L$ by
\[
n_k=\max\left\{\ln(2^{3/4}N_k),\ln\frac{1}{\epsilon}\right\}
\]
implying that $(2^{3/4}N_k)^{\frac1{n_k}}\leq\e$. Choosing $n=n_l$, $p=\ln(2^{3/4}K^2s)$ in lemma \ref{lem:auxest}, and using that $\|\cdot\|_{\boldsymbol{\eps},2n_l}\leq\|\cdot\|_{\boldsymbol{\eps},2n_k}, k\geq l$ we get after this step the estimates
\begin{eqnarray*}
S_l&\leq&\left(\frac{2^{10-2k}K^2spn_l}{m}\right)^{1/2}S^{1/q}=:A_lS^{1/q}\\
\|f_k\|_{\boldsymbol{\eps},2n}&\leq&\|f_k\|_{\boldsymbol{\eps},2n_k}\leq\left(\frac{2^{10-2k}K^2spn_k}{m}\right)^{1/2}S^{1/q}=:A_kS^{1/q}, l<k\leq L.
\end{eqnarray*}
Then by the triangle inequality we have shown
\[
S_l+S_{l+1,L}\leq\sum_{k=l}^L A_kS^{1/q}=\left(\frac{2^{10}K^2sp}{m}\right)^{1/2}S^{1/q}\sum_{k=l}^L\sqrt{2^{-2k}n_k}.
\]
Introducing the covering numbers $N_k$ from the remark after lemma \ref{lem:c1} and observing that $l\geq\frac12\log_2(2^5p)$, we have that if $N\geq4p$ (true by assumption)
\begin{multline*}
2^{3/4}N_k=2^{3/4}\left(\frac{2N\e}{2^{2k}}+\e\right)^{2^{2k}}\leq\left(2^{3/(4\cdot2^{2l})}\left(\frac{2N\e}{2^{2l}}+\e\right)\right)^{2^{2l}}\\
\left(\frac{2^{3/(2^7p)}\e N}{p}\left(\frac1{16}+\frac{p}{N}\right)\right)^{2^{2k}}\leq\left(\frac{N}{p}\right)^{2^{2k}}.
\end{multline*}
This implies that
\begin{multline}
\sum_{k=l}^L\sqrt{2^{-2k}n_k}=\sum_{k=l}^L\sqrt{2^{-2k}\max\{\ln(\sqrt2N_k),\ln(1/\epsilon)\}}\leq\\
\sum_{k=l}^L\max\{\ln^{1/2}(N/p),2^{-k}\ln^{1/2}(1/\epsilon)\}\leq\\(L-l+1)\ln^{1/2}(N/p)+\frac{\ln^{1/2}(1/\epsilon)}{2^{l-1}}\label{eq:sumest}
\end{multline}
To get a bound on $L$ we use the bound of $S_{L+1}$ given by lemma \ref{lem:auxest}. The right hand side of \eqref{est:S_L+1}, and hence also $S_{L+1}$, is less than or equal to $\frac{\lambda g S^{1/q}}{2}$ if and only if
\[
L\geq\frac12\log_2\left(\frac{2^9K^2sp}{(\lambda g)^2}\right),
\]
so we choose 
\[
L=\left\lceil\frac12\log_2\left(\frac{2^9K^2sp}{(\lambda g)^2}\right)\right\rceil.
\]
By the above estimates on $l$ and $L$ we get 
\begin{align*}
L-l+1\leq\frac12\log_2\left(\frac{2^9K^2sp}{(\lambda g)^2}\right)-\frac12\log_2\left(\frac{2^9p}{\e}\right)+3=\frac1{2\ln2}\ln\left(\frac{2^6\e K^2s}{(\lambda g)^2}\right)\\
2^{1-l}\leq\left(\frac{\e}{2^5p}\right)^{1/2}
\end{align*}
Plugging this into \eqref{eq:sumest}, we have shown
\[
\sum_{k=l}^{L}\sqrt{2^{-2k}n_k}\leq\frac1{2\ln2}\ln\left(\frac{2^6\e K^2s}{(\lambda g)^2}\right)\ln^{1/2}(N/p)+\left(\frac{\e}{2^5p}\right)^{1/2}\ln^{1/2}(1/\epsilon).
\]
Thus
\begin{multline*}
S_l+S_{l+1,L}\leq\\
\left(\frac{2^{8}K^2s}{(\ln2)^2m}\right)^{1/2}\left(p^{1/2}\ln\left(\frac{2^6\e K^2s}{(\lambda g)^2}\right)\ln^{1/2}(N/p)+\left(\frac{\e\ln2}{2^3}\right)^{1/2}\ln^{1/2}(1/\epsilon)\right)S^{1/q}.
\end{multline*}
Set now
\begin{equation*}
H=\left(\frac{2^{8}K^2s}{(\ln2)^2m}\right)^{1/2}\left(p^{1/2}\ln\left(\frac{2^6\e K^2s}{(\lambda g)^2}\right)\ln^{1/2}(N/p)+\ln^{1/2}(1/\epsilon^\alpha)\right),\ \alpha=\frac{\e\ln2}{8},
\end{equation*}
so that what we have shown can be expressed by
\[
E_{2n}=S_l+S_{l+1,L}+S_{L+1}\leq\frac{HS^{1/q}}{2}+\frac{\lambda g S^{1/q}}{2}=\frac{H+\lambda g}{2}S^{1/q}.
\]
Plugging stochastic rows $X_j$ back in $S=S(X)$ we have shown
\begin{multline*}
\EV{}\delta_s^{2n}=\EV{}[(2E_{2n}(X))^{2n}]\le( H+\lambda g)^{2n}\EV{}S^{2n/q}=\\
(H+\lambda g)^{2n}\EV{}\sup_{\bu\in\U}\left(\frac1m\sum_{j=1}^m|\langle X_j,\bu\rangle|^2\right)^{n/q}=\\
(H+\lambda g)^{2n}\EV{}\sup_{\bu\in\U}\left(\frac1m\sum_{j=1}^m|\langle X_j,\bu\rangle|^2-\|\bu\|_2^2+\|\bu\|_2^2\right)^{n/q}\leq\\ 
(H+\lambda g)^{2n}\EV{}[(\delta_s+1)^{n/q}].
\end{multline*}
This finally implies that
\begin{equation}
\EV{}[\delta_s^{2n}]^{1/n}\leq (H+\lambda g)^2(\EV{}[\delta_s^{n/q}]^{q/n}+1)^{1/q}\leq (H+\lambda g)^2(\EV{}[\delta_s^{2n}]^{1/2n}+1)^{1/q},
\end{equation}
which concludes the proof of lemma \ref{lem:cru}.
\end{proof}



\section{Appendix}
\subsection{Proof of theorem \ref{thm:best_delta}}
The proof of this theorem requires some simple lemmas.
\begin{lemma}
\label{lem:tlem}
Let $A$ be an $m\times N$-matrix satisfying the RIP-estimate with constants $\delta_s$ and $\bx,\by\in\C^N$ be vectors such that $|\supp\bx\cup\supp\by|\leq2s$ and $\langle\bx,\by\rangle=0$. Let $|t|\leq1$ be such that
\[
\|A\bx\|_2^2-\|\bx\|_2^2=t\delta_{2s}\|\bx\|_2^2
\]
then,
\[
|\langle A\bx,A\by\rangle|\leq\delta_{2s}\sqrt{1-t^2}\|\bx\|_2\|\by\|_2.
\]
\end{lemma}
\begin{proof} We can assume $\|\bx\|_2=\|\by\|_2=1$. Pick $\alpha\geq0, \beta\geq0,\gamma=\pm1$ and consider vectors $\alpha\bx+\gamma\by$ and $\beta\bx-\gamma\by$, then
\begin{align}
\|\alpha\bx+\gamma\by\|_2^2=\alpha^2+1\label{eq:normexp1}\\
\|\beta\bx-\gamma\by\|_2^2=\beta^2+1\label{eq:normexp2}\\
\|A(\alpha\bx+\gamma\by)\|_2^2=\alpha^2\|A\bx\|_2^2+\|A\by\|_2^2+2\alpha\gamma\langle A\bx,A\by\rangle\label{eq:normexp3}\\
\|A(\beta\bx-\gamma\by)\|_2^2=\beta^2\|A\bx\|_2^2+\|A\by\|_2^2-2\beta\gamma\langle A\bx,A\by\rangle\label{eq:normexp4}
\end{align}
Furthermore since $A$ satisfies the restricted isometry property
\begin{align}
\|A(\alpha\bx+\gamma\by)\|_2^2-\|\alpha\bx+\gamma\by\|_2^2\leq\delta_{2s}\|\alpha\bx+\gamma\by\|_2^2\label{eq:RIPest1}\\
\|A(\beta\bx-\gamma\by)\|_2^2-\|\beta\bx-\gamma\by\|_2^2\leq\delta_{2s}\|\beta\bx-\gamma\by\|_2^2.\label{eq:RIPest2}
\end{align}
Subtracting \eqref{eq:RIPest2} from \eqref{eq:RIPest1} and plugging in \eqref{eq:normexp1}-\eqref{eq:normexp4} we get
\begin{align*}
(\alpha^2-\beta^2)\|A\bx\|_2^2+2\gamma(\alpha+\beta)\langle A\bx,A\by\rangle-\alpha^2+\beta^2\leq\delta_{2s}(\alpha^2+\beta^2+2)\iff\\
2\gamma(\alpha+\beta)\langle A\bx,A\by\rangle\leq(\beta^2-\alpha^2)(\|A\bx\|_2^2-\|\bx\|_2^2)+\delta_{2s}(\alpha^2+\beta^2+2)\iff\\
\gamma\langle A\bx,A\by\rangle\leq\delta_{2s}\frac{\alpha^2(1-t)+\beta^2(1+t)+2}{2(\alpha+\beta)}.
\end{align*}
Since this holds for $\gamma=\pm1$ and if we set $f(\alpha,\beta)=\frac{\alpha^2(1-t)+\beta^2(1+t)+2}{2(\alpha+\beta)}$ we have shown
\[
|\langle A\bx,A\by\rangle|\leq\delta_{2s}f(\alpha,\beta).
\]
Finally we find the minimum value of $f$ in the first quadrant to be $\sqrt{1-t^2}$ at the critical point $(\alpha,\beta)=\left(\sqrt{\frac{1+t}{1-t}},\sqrt{\frac{1-t}{1+t}}\right)$. Hence
\[
|\langle A\bx,A\by\rangle|\leq\delta_{2s}\sqrt{1-t^2}\|\bx\|_2\|\by\|_2.
\]
\end{proof}
The following result can be found in \cite{CWX2} (proposition 2.1).
 \begin{lemma}\label{lem:CWX2}
Suppose $\bx=(x_1,x_2,\dots,x_s), x_1\geq x_2\geq\cdots\geq x_s\geq0$, then
\[
\|\bx\|_2\leq\frac1{\sqrt{s}}\|\bx\|_1+\frac{\sqrt{s}}4(x_1-x_s).
\]
\end{lemma}
\begin{lemma}
\label{lem:l21}
Assume $\mathbf{x}=(x_1,\dots,x_N)\in\C^N$ such that $|x_1|\geq|x_2|\geq\dots\geq|x_N|$. Write $\mathbf{x}=\sum_k\mathbf{x}_{S_k}$ where $S_1=\{1,\dots,t\},S_k=\{t+(k-2)s+1,\dots,t+(k-1)s\}, k>1$, so that $|S_1|=t, |S_k|=s, k>1$ (except for possibly the last $k$), then
\[
\sum_{k>1}\|\mathbf{x}_{S_k}\|_2\leq\frac1{\sqrt s}\|\bx_{S_1^c}\|_1+\frac{\sqrt{s}}4|x_{s+1}|
\]
\end{lemma}
\begin{proof}
\begin{multline*}
\|\mathbf{x}_{S_k}\|_2\leq\frac1{\sqrt s}\|\mathbf{x}_{S_{k}}\|_1+\frac{\sqrt s}4(|x_{s+(k-2)t+1}|-|x_{s+(k-1)t}|)\leq\\
\frac1{\sqrt s}\|\mathbf{x}_{S_{k}}\|_1+\frac{\sqrt s}4(|x_{s+(k-2)t+1}|-|x_{s+(k-1)t+1}|)
\end{multline*}
by lemma \ref{lem:CWX2}. Summing this over $k>1$ gives (since $S_k\cap S_l=\emptyset, k\neq l$)
\begin{multline*}
\sum_{k>1}\|\mathbf{x}_{S_k}\|_2\leq\sum_{k>1}\left(\frac1{\sqrt s}\|\mathbf{x}_{S_{k}}\|_1+\frac{\sqrt s}4(|x_{s+(k-2)t+1}|-|x_{s+(k-1)t+1}|)\right)\leq\\
\frac1{\sqrt s}\|\bx_{S_1^c}\|_1+\frac{\sqrt{s}}4|x_{s+1}|
\end{multline*}
\end{proof}
\begin{proof}[Proof of proposition \ref{prop:nsp2}] The proposition follows by lemma \ref{lem:l21} if with $t=s$ since then we can estimate the last term in the inequality with
\[
|x_{s+1}|\leq\frac1s\|\mathbf{x}_{S_1}\|_1.
\]
Using this one gets
\begin{gather*}
\frac{\|\bx_{S_1}\|_1}{\sqrt{s}}\leq\|\bx_{S_1}\|_2<\frac45\sum_{k>1}\|\mathbf{x}_{S_k}\|_2\leq\frac1{5\sqrt{s}}\|\mathbf{x}_{S_1}\|_1+\frac4{5\sqrt{s}}\|\mathbf{x}_{S_1^c}\|_1\implies\\
\|\bx_{S_1}\|_1<\|\bx_{S_1^c}\|_1.
\end{gather*}
It is now clear that the same holds for any subset $S\subset[N]$ with $|S|=s$.
\end{proof}
\begin{proof}[Proof of theorem \ref{thm:best_delta}]
Take $A$ and $t$  as in lemma \ref{lem:tlem} and $\bx=\{\bx_{S_k}\}$ as in lemma \ref{lem:l21} (with $|S_k|=s, k=1,2,\dots$, except for possibly the last $k$) such that $A\bx=0$. Then we get since $\|A\bx_{S_1}\|_2^2\geq(1-t\delta_{2s})\|\bx_{S_1}\|_2^2$ that
\begin{multline*}
(1-t\delta_{2s})\|\bx_{S_1}\|_2^2\leq\|A\bx_{S_1}\|_2^2\leq\langle A\bx_{S_1},-A\bx_{S_1^c}\rangle\leq\sum_{k>1}\langle A\bx_{S_1},A(-\bx_{S_k})\rangle\\
\leq\delta_{2s}\sqrt{1-t^2}\|\bx_{S_1}\|_2\sum_{k>1}\|x_{S_k}\|_2\iff\|\bx_{S_1}\|_2\leq\frac{\delta_{2s}\sqrt{1-t^2}}{1-t\delta_{2s}}\sum_{k>1}\|x_{S_k}\|_2.
\end{multline*}
Now we use lemma \ref{lem:l21} and the inequality $\|\bx_{S_1}\|_1\leq\sqrt{s}\|\bx_{S_1}\|_2$.
\begin{gather*}
\frac{\|\bx_{S_1}\|_1}{\sqrt{s}}\leq\frac{\delta_{2s}\sqrt{1-t^2}}{1-t\delta_{2s}}\frac{1}{\sqrt{s}}\left(\|\mathbf{x}_{S_1^c}\|_1+\frac14\|\mathbf{x}_{S_1}\|_1\right)\iff\\
\|\bx_{S_1}\|_1\left(1-\frac{\delta_{2s}\sqrt{1-t^2}}{4(1-t\delta_{2s})}\right)\leq\|\bx_{S_1^c}\|_1\frac{\delta_{2s}\sqrt{1-t^2}}{1-t\delta_{2s}}.
\end{gather*}
It follows that $\|\bx_{S_1}\|_1<\|\bx_{S_1^c}\|_1$ (i.e. the null space property is fulfilled) if
\begin{align*}
\left(1-\frac{\delta_{2s}\sqrt{1-t^2}}{4(1-t\delta_{2s})}\right)>\frac{\delta_{2s}\sqrt{1-t^2}}{1-t\delta_{2s}}\iff\frac45>\frac{\delta_{2s}\sqrt{1-t^2}}{1-t\delta_{2s}}.
\end{align*}
Now observe that the minimum of the right hand side is attained at $t=\delta_{2s}$ and hence we want
\[
\frac45>\frac{\delta_{2s}}{\sqrt{1-\delta_{2s}^2}},
\]
which is fulfilled as long as $\delta_{2s}<\frac4{\sqrt{41}}$.
\end{proof}
\subsection{Improving on theorem \ref{thm:best_delta}}
\label{sec:improvedelta}
Here we will sketch out the details for our, so far, best improvement of theorem \ref{thm:best_delta}. The first step involves to replace $t$ and $s$ in lemma \ref{lem:l21} with $\lceil6s/5\rceil$ and $\lfloor4s/5\rfloor$ respectively, where $s\geq2$ is an integer. We also introduce $S=\{1,2,\dots,s\}\subset S_1$. Then we have
\begin{multline}
\sum_{k>1}\|\mathbf{x}_{S_k}\|_2\leq\frac1{\sqrt{\lfloor4s/5\rfloor}}\|\bx_{S_1^c}\|_1+\frac{\sqrt{\lfloor4s/5\rfloor}}4|x_{\lceil6s/5\rceil+1}|=\\
\frac1{\sqrt{\lfloor4s/5\rfloor}}\left(\|\bx_{S^c}\|_1-\left(\|\bx_{S_1\setminus S}\|_1-\frac{\lfloor4s/5\rfloor}4|x_{\lceil6s/5\rceil+1}|\right)\right)\leq\frac1{\sqrt{4s/5-1}}\|\bx_{S^c}\|_1.\label{helpcalc}
\end{multline}
The last inequality follows since
\begin{multline*}
\|\bx_{S_1\setminus S}\|_1-\frac{\lfloor4s/5\rfloor}4|x_{\lceil6s/5\rceil+1}|\geq(\lceil6s/5\rceil-s)|x_{\lceil6s/5\rceil+1}|-\frac{\lfloor4s/5\rfloor}4|x_{\lceil6s/5\rceil+1}|=\\
\left(\lceil6s/5\rceil-s-\frac{\lfloor4s/5\rfloor}4\right)|x_{\lceil6s/5\rceil+1}|\geq0.
\end{multline*}
Observe that if $5$ divides $s$, we may replace $\frac1{\sqrt{4s/5-1}}$ with $\sqrt{\frac54}$. The improvement of proposition \ref{prop:nsp2} becomes:
\begin{proposition}
Assume $\bx=(x_1,\dots,x_N)\in\C^N$ such that $|x_1|\geq|x_2|\geq\dots\geq|x_N|$ and that $s\geq2$ is an integer. Write $\bx=\sum_k\bx_{S_k}$ where $S_1=\{1,\dots,\lceil6s/5\rceil\},S_2=\{\lceil6s/5\rceil+1,\dots,\lceil6s/5\rceil+\lfloor4s/5\rfloor\}$ etc. so that $|S_1|=\lceil6s/5\rceil, |S_k|=\lfloor4s/5\rfloor, k\geq2$ (except for possibly the last $k$).Then if
\[
\|\bx_{S_1}\|_2<\frac1{\sqrt{s}}\sqrt{4s/5-1}\sum_{k>1}\|\bx_{S_k}\|_2,
\]
it holds that $\|\bx_S\|_1<\|\bx_{S^c}\|_1$ for all subsets $S\subset[N]$ with $|S|=s$. In particular if $5$ divides $s$, 
\[
\|\bx_{S_1}\|_2<\sqrt{\frac45}\sum_{k>1}\|\bx_{S_k}\|_2
\]
is sufficient.
\end{proposition}
\begin{proof}
If $S=\{1,\dots,s\}\subset S_1$, then by \eqref{helpcalc}
\[
\frac{\|\bx_S\|_1}{\sqrt{s}}\leq\|\bx_S\|_2\leq\|\bx_{S_1}\|_2<\frac1{\sqrt{s}}\sqrt{4s/5-1}\sum_{k>1}\|\bx_{S_k}\|_2\leq\frac1{\sqrt{s}}\|\bx_{S^c}\|_1.
\]
\end{proof}
Now we can simply modify the proof of theorem \ref{thm:best_delta} in the previous section in a rather obvious way to find that
\[
\delta_{2s}<\begin{cases}
\sqrt{\frac{4-5/s}{9-5/s}}&, 2\leq s, 5\text{ does not divide } s\\
\frac23&, 2\leq s, 5\text{ divides } s
\end{cases}
\]
implies that the matrix $A$ with restricted isometry constants $\delta_s$ satisfies the null space property of order $s$. The combination of the result of theorem \ref{thm:best_delta} (which is better for small $s$) with the improved one above can be summarized in the following figure
\begin{figure}[h!]
\centering
\includegraphics[scale=0.63]{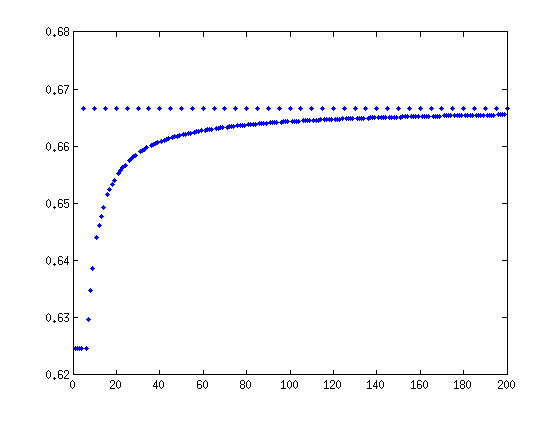}
\caption{Plot of optimal bounds of the constants $\delta_{2s}$ for $s=1,\dots,200$, implying NSP. For the smallest $s$, $4/\sqrt{41}$ is best, while if $5$ divides $s$, $2/3$ will do. For larger $s$ that is not divisible by $5$ an upper bound is given by  $\sqrt{\frac{4-5/s}{9-5/s}}$.}
\label{fig:bestdelta}
\end{figure}


\bibliography{mybib}{}
\bibliographystyle{plain}

\end{document}